\documentclass[journal,10pt]{IEEEtran}
\usepackage{amssymb}
\usepackage{amsmath}
\usepackage{cite}
\usepackage{url}
\usepackage{xcolor}
\usepackage{cite,graphicx,amsmath,amssymb}
\usepackage{fancyhdr}
\usepackage{mdwmath}
\usepackage{mdwtab}
\usepackage{caption}
\usepackage{amsthm}
\usepackage{setspace}
\usepackage{algorithm}
\usepackage{algorithmic}
\usepackage{makecell}
\usepackage{diagbox}

\usepackage{stfloats}
\usepackage{subfig}%多个子图
\usepackage{float}

\newtheorem{theorem}{Theorem}

\newtheorem{lemma}{Lemma}

\newtheorem{corollary}{Corollary}

\newtheorem{proposition}{Proposition}
\allowdisplaybreaks
\makeatletter
\def\ScaleIfNeeded{%
\ifdim\Gin@nat@width>\linewidth \linewidth \else \Gin@nat@width
\fi } \makeatother

\begin{document}
%\title{Multiple Access for Deployment-aware Intelligent Reflecting Surface Assisted Networks}
%\title{Multiple Access for Intelligent Reflecting Surface Assisted Networks: A Deployment Optimization Perspective}
\title{Channel Estimation for STAR-RIS-aided \\Wireless Communication}
%
%Trajectory Design for Mission Completion Time Minimization in Cellular-enabled UAV Uplink NOMA Communication
%\author{
%\IEEEauthorblockN{ Yuanwei~Liu\IEEEauthorrefmark{1}, Zhijin~Qin\IEEEauthorrefmark{1}, Maged Elkashlan\IEEEauthorrefmark{1}, and  Yue~Gao\IEEEauthorrefmark{1}\\} \IEEEauthorblockA{
%\IEEEauthorrefmark{1} Queen Mary University of London, London, UK\\
%%\IEEEauthorrefmark{2} Lancaster University, Lancaster, UK\\
% } }

\author{
	Chenyu~Wu, Changsheng~You, \IEEEmembership{Member, IEEE}, Yuanwei~Liu, \IEEEmembership{Senior Member,~IEEE}, Xuemai Gu, \IEEEmembership{Member,~IEEE}, and Yunlong~Cai, \IEEEmembership{Senior Member,~IEEE}	\vspace{-2em}
	
	%\IEEEauthorblockN{,~\IEEEmembership{Member,~IEEE,}}
	% ,~\IEEEmembership{Fellow,~IEEE,}
	%and \IEEEauthorblockN{,~\IEEEmembership{Senior Member,~IEEE,}     \IEEEmembership{Graduate Student Member, IEEE}

\thanks{C. Wu, and X. Gu are with the School of Electronic and Information Engineering, Harbin Institute of Technology (HIT), Harbin 150001, China (e-mail: \{wuchenyu, guxuemai\}@hit.edu.cn).}
\thanks{C. You is with the Department of Electrical and Electronic Engineering, Southern University
of Science and Technology (SUSTech), Shenzhen 518055, China (email: youcs@sustech.edu.cn). %he was with the Department of Electrical and Computer Engineering, National University of Singapore, Singapore 117583.
}
\thanks{Y. Liu is with the School of Electronic Engineering and Computer Science, Queen Mary University of London, London E1 4NS, UK (email: yuanwei.liu@qmul.ac.uk).}

\thanks{Y. Cai is with the College of Information Science and Electronic Engineering, Zhejiang University, Hangzhou 310027, China (email: ylcai@zju.edu.cn).}
%\thanks{O. A. Dobre is with the Department of Electrical and Computer Engineering, Memorial University, St. John’s, NL A1C 5S7, Canada (e-mail: odobre@mun.ca).}
}

\maketitle

\begin{abstract}
In this letter, we study efficient uplink channel estimation design for a simultaneously transmitting and reflecting reconfigurable intelligent surface (STAR-RIS) assisted two-user communication systems. We first consider the time switching (TS) protocol for STAR-RIS and propose an efficient scheme to separately estimate the channels of the two users with optimized training (transmission/reflection) pattern. Next, we consider the energy splitting (ES) protocol for STAR-RIS under the practical coupled phase-shift model and devise a customized scheme to simultaneously estimate the channels of both users. Although the problem of minimizing the resultant channel estimation error for the ES protocol is difficult to solve, we propose an efficient algorithm to obtain a high-quality solution by jointly designing the pilot sequences, power-splitting ratio, and training patterns. Numerical results show the effectiveness of the proposed channel estimation designs and reveal that the STAR-RIS under the TS protocol achieves a smaller channel estimation error than the ES case.

%Two protocols for operating STAR-RIS (namely, time switching (TS) and energy splitting (ES)) are presented, based on which the two users' channels are consecutively and simultaneously estimated, respectively. %where two users located on each side of the STAR-RIS communicate with a base station.
%Particularly, the ideal and practical phase-shift models of the STAR-RIS for ES are both considered. %where the phase adjustment for transmission and reflection is independent and coupled, respectively. 
%For TS, we demonstrate that the channel of one typical user can be estimated optimally with the aid of properly designed activation patterns of the STAR-RIS. For ES, to minimize the sum mean square error (MSE) of channel estimation, an optimization problem of jointly designing the pilot sequence of the users and the transmission- and reflection-pattern of the STAR-RIS is formulated. To facilitate solving it, we propose an optimal scheme for the ideal phase-shift case and a high-quality suboptimal solution for the practical phase-shift case. Numerical results show that 1) With the same overhead, TS is more cost-effective in terms of uplink channel estimation, compared with ES. 2) Aided by our design, the channel estimation accuracy for ES with practical phase-shifts is nearly the same as that with ideal phase-shifts.
%the MSE of channel estimation employing ES almost doubles that employing
%TS. 

%phase-shift coefficients for transmission and reflection can be independently adjusted
\end{abstract}

\begin{IEEEkeywords}
	Channel estimation,	reconfigurable intelligent surface, simultaneous transmission and reflection. 
\end{IEEEkeywords}

\section{Introduction}
%\IEEEPARstart{R}{ecently}, 
Recently, reconfigurable intelligent surface (RIS) has emerged as a  promising technology to improve the spectrum and energy efficiency of future wireless systems\cite{survey_ris}. Specifically, RIS enables reconfigurable radio environment by smartly tuning the signal propagation via a large number of low-cost elements. However, in most existing works, RISs can only reflect signals within its front half-space\cite{zr_survey}, thus can only serve users on one side. To overcome this limitation, a novel concept of simultaneously transmitting and reflecting RISs (STAR-RISs) has been recently proposed\cite{STAR}, which can transmit and/or reflect the incident signals, thus enabling serving users in both sides of RIS to achieve \emph{full-space} smart radio environment. %Moreover, STAR-RISs provide additional degrees-of-freedoms (DoF) for signal manipulation, which enhances the flexibility of network design. 
This has motivated growing research interests to investigate the benefits of deploying STAR-RISs in wireless networks, such as extending network coverage\cite{star2}, reducing power consumption\cite{star1}, and enhancing system throughput\cite{ios}.

To reap the passive beamforming gain of RISs/STAR-RISs, channel state information (CSI) is indispensable, which, however, is practically challenging to acquire due to the lack of signal transmission/processing capabilities for passive surface. To address this issue, various RIS channel estimation schemes have been proposed in the literature to attain the CSI by efficiently turning the reflections of RIS over time based on e.g., the ON/OFF-based\cite{on_off} and discrete Fourier transform (DFT)-based\cite{dft} training reflection pattern. %To acquire CSI of STAR-RIS-aided network, a simple idea is to periodically switch the STAR-RIS between reflecting-only and transmitting-only mode and consecutively estimate the channels of users located at different sides using existing methods\cite{on_off,dft}. However, this \emph{time switching} scheme entails high hardware cost \cite{star1} and can not fully exploit the potential of the STAR-RIS during downlink transmission since the DoF for network design is limited. 
%To make full use the flexibility of the STAR-RIS, we first propose an uplink \emph{energy splitting} model to enable pilot signals from both sides of the surface to reach the BS. 
However, these schemes \cite{on_off,dft} cannot be directly applied to STAR-RIS due to the following reasons. 
First, the channel estimation design for STAR-RIS hinges on its hardware design and operation protocol, e.g., time switching (TS) and energy splitting (ES) \cite{STAR,star1}; Though ES is shown to be preferable for downlink data transmission\cite{star1}, it is unknown whether ES maintains the superiority over TS in terms of channel estimation. Thus, new channel estimation schemes tailed to STAR-RIS need to be devised. Second, in contrast to the conventional RIS that requires designing a single reflection pattern only, both the transmission and reflection patterns of the STAR-RIS need to be jointly designed with the pilot sequences. In particular, for ES, the STAR-RIS phase-shifts for transmission and reflection are intricately coupled in practice, which makes the channel estimation design more complicated. 
%Thus, an efficient channel estimation scheme is necessary to ensure the accuracy of CSI under the practical phase-shift model.

To address the above issues, we study in this letter efficient channel estimation schemes for a STAR-RIS-assisted two-user communication system based on two practical protocols of TS and ES. For TS, we design the optimal training (transmission/reflection) pattern to separately estimate the concatenated channels of the two users. For ES, we consider a practical coupled phase-shift model for STAR-RIS and propose an efficient scheme to simultaneously estimate the channels of both users. Specifically, to minimize the mean square error (MSE) of channel estimation for the ES case, a joint optimization problem of the power-splitting ratio, pilot sequences, and training patterns is formulated. To solve this challenging problem, we first obtain the optimal solution under the ideal case where the transmission-and-reflection patterns are independently controlled, based on which a near-optimal solution under the practical phase-shift model is then developed. Interestingly, it is shown that the channel estimation error for ES is generally larger than the TS case as ES results in power leakage in the uplink channel estimation. Simulation results corroborate the effectiveness of our proposed designs and theoretical analysis.

\emph{Notations:} $\mathbb{C}^{M\times1}$ denotes the space of $M\times1$ complex-valued
vectors. $\mathbf{a}^T$ and $\mathbf{a}^H$ denote the transpose and the conjugate transpose of vector $\mathbf{a}$, respectively. $\text{diag}(\mathbf{a})$
denotes a diagonal matrix with the diagonal elements given by $\mathbf{a}^T$. For a square matrix $\mathbf{A}$, $\text{Tr}(\mathbf{A})$, and $\mathbf{A}^{-1}$ represent its trace and inverse, respectively. For any matrix $\mathbf{B}$, $\text{rank}(\mathbf{B})$ and $[\mathbf{B}]_{m,n}$ denote its rank and $(m,n)$th element. $\jmath$ denotes the imaginary unit, i.e., $\jmath^2=-1$. The distribution
of a circularly symmetric complex Gaussian variable with mean $\mu$ and variance $\sigma^2$ is denoted by $\mathcal{C}\mathcal{N}\sim{(\mu,\sigma^2)}$.

\section{System Model}
%\begin{figure}[]
%	\centering
%	\includegraphics[width=0.5\textwidth]{eps/model.eps}\\
%	\caption{A STAR-RIS aided two-user uplink communication system.}
%\end{figure}

\begin{figure}[t]
	\centering
	\subfloat[A STAR-RIS aided two-user uplink communication systems]{		
		\includegraphics[width=0.23\textwidth]{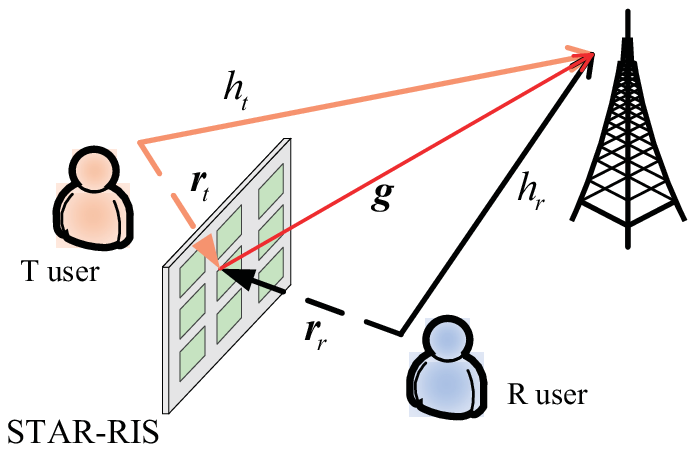}	
	}
	\subfloat[Operation protocols for STAR-RIS]{		
		\includegraphics[width=0.25\textwidth]{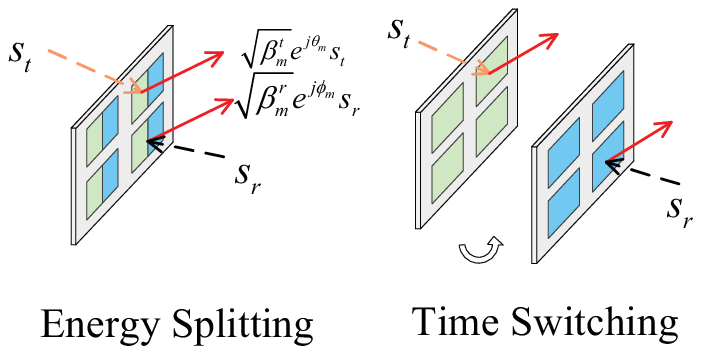}
	}	
	\caption{Illustration of system model.}
\end{figure}

As illustrated in Fig. 1(a), we consider a narrow-band wireless communication system, where a STAR-RIS with $M_0$ elements is deployed to assist the data transmission from two single-antenna users\footnote{Our design can be extended to the multi-user case by dividing users into multiple groups, each consisting of two users located on the two sides of RIS, and allocating orthogonal time/frequency resources to each group.} located on its both sides to a single-antenna base station\footnote{The design in our work is readily extendable to the BS with multiple antennas by estimating their channels in parallel.} (BS). %The full region of signal propagation is divided into two parts by the STAR-RIS, namely, reflection region and transmission region. 
We recall the elements-grouping strategy in \cite{group,group2} to reduce the channel estimation overhead, where the STAR-RIS elements are divided into $M$ sub-surfaces, each of which consists of $M_0/M$ adjacent elements that share a common transmission/reflection coefficient. The uplink signals of the two users can be either transmitted or reflected by the STAR-RIS to the BS, thus referred to as T user and R user, respectively. Let $h_k$ and $\mathbf{r}_k\in\mathbb{C}^{M\times1}$ denote the baseband equivalent channel from user $k$, $k\in\mathcal{K}\triangleq\{t,r\}$ to the BS and STAR-RIS, respectively. Further, we denote $\mathbf{g}^H\in\mathbb{C}^{1\times{M}}$ as the channel from the STAR-RIS to the BS; and $\mathbf{q}_k^H\triangleq\text{diag}(\mathbf{g})\mathbf{r}_k\in\mathbb{C}^{M\times{1}}$ as the cascaded channels from user $k$ to the BS. Moreover, we assume quasi-static block-fading channels and focus on the uplink communication in one typical fading block such that ${h}_k$,  $\mathbf{r}_k$, and $\mathbf{g}$ remain approximately constant. 

We consider two practical operation protocols for STAR-RIS, namely, time switching and energy splitting\cite{STAR,star1}, as shown in Fig. 1(b). In the following sections, we will elaborate the STAR-RIS channel estimation scheme for TS and ES, respectively.

\section{STAR-RIS Channel Estimation: Time Switching}
In this section, we propose an efficient scheme to separately estimate the concatenated channels of STAR-RIS under the TS protocol.  

\underline{Transmission and Reflection Model:} For TS, the STAR-RIS switches all elements between the transmission and reflection modes in two separate time intervals (referred to as T period and R period). 
We denote $\tau_t$ and $\tau_r$ as the number of time slots allocated to the T period and R period, respectively. Moreover, we define $\boldsymbol{\theta}_i^{\text{TS}}\triangleq[{e^{\jmath{\theta_{1,i}}}},..., \ldots ,{e^{\jmath{\theta _{M,i}}}}],\boldsymbol{\phi}_i^{\text{TS}}\triangleq[{e^{\jmath{\phi_{1,i}}}},..., \ldots ,{e^{\jmath{\phi _{M,i}}}}]$ as the transmission and reflection coefficient vectors, respectively, where
%The corresponding transmission and reflection-coefficient vectors are given by $\boldsymbol{\theta}^{TS}= [{e^{j{\theta_{1}}}},{{e^{j{\theta _{2}}}}}, \ldots ,{e^{j{\theta _{M}}}}]$
%and $\boldsymbol{\phi}^{TS}=[{e^{j{\phi _{1}}}},{e^{j{\phi _{2}}}}, \ldots,{e^{j{\phi _{M}}}}]$, 
$\theta_{m,i},\phi_{m,i}\in[0,2\pi),m\in\mathcal{M}\triangleq\{1,...,M\}$ is the phase-shift of the $m$-th sub-surface in time slot $i$. 
%Since we exploit different time domain for transmission and reflection, the phase-shifts $\theta_{m,i},\phi_{m,i}$ can be independently controlled.

\underline{Signal Model:} In the following, we elaborate the channel estimation design for the T user, while the scheme can be directly extended to the case of R user. Specifically, during the T period, the T user consecutively sends pilot symbols to the BS. Denote $p$ as the maximum transmit power.
%The users are treated equally during channel estimation since we do not have any prior information regarding the channel gains. Thus, the estimation period and the transmit power of the users are set to be identical, i.e., $\tau_t=\tau_r,p_t=p_r=p$.
Then, the baseband received signal at the BS in time slot $i$, $i=1,...,\tau_t$, is
\begin{equation}\label{ts_1}
\begin{aligned}
	y_i&=[h_t+\boldsymbol{g}^H\text{diag}(\boldsymbol{\theta}_i^{\text{TS}})\mathbf{r}_t]\sqrt{p}s_{t,i}+n_{i}\\&=(h_t+\boldsymbol{\theta}_i^{\text{TS}}\mathbf{q}_t^H)\sqrt{p}s_{t,i}+n_{i},
\end{aligned}
\end{equation}
where $n_i\sim {\mathcal{C}\mathcal{N}}\left(0,\sigma^2\right)$ is the additive white Gaussian noise at the BS; $|s_{t,i}|$ is the pilot symbol in time slot $i$, which can be set as $s_{t,i}=1$ for simplicity. As such, the overall received signal during the T period at the BS can be written as
%(\ref{ts_1}) can be viewed as a linear system. Since there are $M+1$ unknown channel coefficients in $h_t$ and $\mathbf{q}_t$ to be estimated, at least $\tau_t\geq{M+1}$ time slots are needed. By setting $s_{t,i}=1,i=1,..,\tau_t$, 

\begin{equation}\label{signal_ts}
	\mathbf{y}^t=[y_1,\dots,y_{\tau_t}]^T=\sqrt{p}\boldsymbol{\Theta}\mathbf{x}^t+\mathbf{n},
\end{equation}
where $\boldsymbol{\Theta}$ denotes the transmission pattern matrix, the i-th row of which is given by $[1,\boldsymbol{\theta}_i^{\text{TS}}]$. $\mathbf{x}^t=[h_t,\mathbf{q}_t]^T$ denotes the composite channel vector associated with the T user, and $\mathbf{n}=[n_1,...,n_{\tau_t}]^T$. 
\vspace{-0.3cm}
\subsection{Problem Formulation}
According to (\ref{signal_ts}), if $\boldsymbol{\Theta}$ is of full column rank, the least-square (LS) estimate of $\mathbf{x}^t$ is given by $\hat{\mathbf{x}}=\frac{1}{\sqrt{p}}\boldsymbol{\Theta}^{\dagger}\mathbf{y}^t$,
where $\boldsymbol{\Theta}^{\dagger}=(\boldsymbol{\Theta}^H\boldsymbol{\Theta})^{-1}\boldsymbol{\Theta}^H$ is the pseudo-inverse of $\boldsymbol{\Theta}$. The MSE of the above LS estimation is\cite{dft}
\begin{equation}
	\begin{aligned}
		\text{MSE}^{\text{T}}=\mathbb{E}\left[\Vert\hat{\mathbf{x}}-\mathbf{x}\Vert_2^2\right]=
		\frac{\sigma^2}{p}\text{Tr}[(\boldsymbol{\Theta}^H\boldsymbol{\Theta})^{-1}].
	\end{aligned}
\end{equation}
The optimization problem for minimizing the MSE can be thus formulated as
\begin{subequations}\label{pro_ts}
	\begin{align}
		{\rm(P1)}:\mathop {\min }\limits_{\boldsymbol{\Theta}} \;\;&\frac{\sigma^2}{p}\text{Tr}[(\boldsymbol{\Theta}^H{\boldsymbol{\Theta}})^{-1}]\\
		\label{unit_mo1}{\rm{s.t.}}\;\;&\theta_{m,i},\phi_{m,i}\in[0,2\pi),m\in\mathcal{M},i=1,...,\tau_t,\\
		\label{rank1}&\text{rank}(\boldsymbol{\Theta})=M+1.
	\end{align}
\end{subequations}
%where constraint (\ref{unit_pilot1}) is to ensure the estimation of the direct link $h_t$, (\ref{unit_mo1}) is the unit-modulus constraint of phase-shifts, 

\subsection{Proposed Solution}
%The optimal solution of problem (\ref{pro_ts}) is obtained when $\boldsymbol{\Theta}$ is an orthogonal matrix\cite{dft}. 
According to \cite{dft}, an optimal solution to problem (P1) is an  ${(M+1)}\times{(M+1)}$ DFT matrix $\mathbf{D}_{M+1}$, whose entries are given by
\begin{equation}
	[\mathbf{D}_{M+1}]_{m,n}=e^{-\jmath\frac{2\pi{(m-1)(n-1)}}{M+1}},1\leq{m,n}\leq{M+1}.
\end{equation}
%Besides, another optimal choice of $\boldsymbol{\Theta}$ is the stack of the first column and other $M$ randomly selected columns in $\mathbf{D}_{\tau_t}$. 
%As a result, the MSEs for estimating direct link and cascaded links are
%\begin{equation}
%	\epsilon({h_t})=\frac{\sigma^2}{p(M+1)},\;\;\epsilon({\mathbf{q}_t})=\frac{M\sigma^2}{p(M+1)}
%\end{equation}
\newcounter{mytempeqncnt}
\setcounter{mytempeqncnt}{\value{equation}}
\begin{figure*}[b]
	% ensure that we have normalsize text
	\normalsize
	\setcounter{equation}{9}
	% Store the current equation number.
	% Set the equation number to one less than the one
	% desired for the first equation here.
	% The value here will have to changed if equations
	% are added or removed prior to the place these
	% equations are referenced in the main text.
	\hrulefill
	\begin{equation}	\label{H_2}
		\mathbf{V}=\begin{gathered}
			\begin{bmatrix}
				s_{t,1} &\sqrt{{\beta _1^t}}{e^{j{\theta _{1,1}}}}s_{t,1}&...&\sqrt{{\beta _M^t}}{e^{j{\theta _{M,1}}}}s_{t,1}& s_{r,1}&\sqrt{{\beta _1^r}}{e^{j{\phi _{1,1}}}}s_{r,1}&...&\sqrt{{\beta _M^r}}{e^{j{\phi _{M,1}}}}s_{r,1} \\
				\vdots&\vdots&\vdots&\ddots&\vdots&\vdots&\ddots&\vdots\\
				s_{t,\tau}&\sqrt{{\beta _1^t}}{e^{j{\theta _{1,\tau}}}}s_{t,\tau}&...&\sqrt{{\beta _M^t}}{e^{j{\theta _{M,\tau}}}}s_{t,\tau}&  s_{r,\tau}&\sqrt{{\beta _1^r}}{e^{j{\phi _{1,\tau}}}}s_{r,\tau}&...&\sqrt{{\beta _M^r}}{e^{j{\phi _{M,\tau}}}}s_{r,\tau}
			\end{bmatrix}
		\end{gathered}	
	\end{equation}
	% Restore the current equation number.
	% IEEE uses as a separator
	%\hrulefill
	% The spacer can be tweaked to stop underfull vboxes.
	\vspace*{4pt}
\end{figure*}
\setcounter{equation}{\value{mytempeqncnt}}

Following the same procedure in the R period, the CSI of the R user can be acquired. Then, it can be easily shown that under the minimum required overhead $\tau_t+\tau_r=2M+2$, the sum MSE of channel estimation by TS is  
\begin{equation}\label{mse_ts}	
	\text{MSE}^{\text{TS}}=\frac{2\sigma^2}{p}.
\end{equation}

\section{STAR-RIS Channel Estimation: Energy Splitting}
In this section, we introduce the channel estimation scheme for ES, where the channels of both users are estimated simultaneously.

\underline{Transmission and Reflection Model:} For the ES protocol, the signals impinged on each sub-surface is split into transmitted ones and reflected ones with an energy splitting ratio of $\beta_m^t,\beta_m^r$. Denote $\tau$ as the number of time slots for channel estimation. Accordingly, the transmission and reflection vectors at time slot $i$ is defined as 
$\boldsymbol{\theta}_i^{\text{ES}}\triangleq[\beta_{1}^t{e^{\jmath{\theta_{1,i}}}}, \ldots ,\beta_{M}^t{e^{\jmath{\theta _{M,i}}}}]$ and $\boldsymbol{\phi}_i^{\text{ES}}\triangleq[\beta_{1}^r{e^{\jmath{\phi_{1,i}}}}, \ldots ,\beta_{M}^r{e^{\jmath{\phi _{M,i}}}}]$, respectively. Note that according to the law of energy conservation, we have $\beta_m^t+\beta_m^r\leq1$\cite{STAR}.

%Different from the ideal phase-shift model (i.e., the phase-shift $\theta_{m,i},\phi_{m,i}$ for transmission and reflection can be adjusted independently) considered in 

We consider in this paper the \emph{practical coupled phase-shift model}, where the phase-shifts of each element for transmission and reflection, i.e., $\theta_{m,i},\phi_{m,i}$ are coupled with each other.\footnote{We assume continuous phase-shifts in this letter, while the results can be extended to a more general case with practical discrete phase-shift by proper quantization\cite{discrete_jsac}.} This model is practically accurate for a fully-passive STAR-RIS under the hardware constraint\cite{2014Dynamic}. Specifically, according to \cite{2014Dynamic,coupled_phase}, the coupled phase-shifts of each sub-surface $m$ should satisfy
\begin{equation}\label{con_couple}
		\cos({\theta_{m,i}}-{\phi_{m,i}})=0.
\end{equation}

\underline{Signal Model:} During the channel estimation stage, both users keep sending pilot symbols to the BS, while the training (transmission and reflection) patterns of the STAR-RIS are properly designed to assist the channel estimation. For a fair comparison, we set the same total transmit power of the users in each time slot as in the TS case; thus, the transmit power of the T user and R user is given by $p_t=p_r=\frac{p}{2}$.
As such, the baseband received signal at the BS in time slot $i$ is 
\begin{equation}
	y_i=(h_t+\boldsymbol{\theta}_i^{\text{ES}}\mathbf{q}_t^H)\sqrt{p/2}s_{t,i}+(h_r+\boldsymbol{\phi}_i^{\text{ES}}\mathbf{q}_r^H)\sqrt{p/2}s_{r,i}+n_{i},
\end{equation}
where $s_{k,i}$ denotes the pilot symbol of user $k$ in time instant $i$. 
%Since there as $2M+2$ unknowns channel coefficients, the training overhead for ES is no less than $2M+2$. 
Then, the overall received signal at the BS during the channel estimation stage is given by

\begin{equation}\label{signal_es}
	\mathbf{y}=[y_1,\dots,y_\tau]^T=\sqrt{p/2}\mathbf{V}\mathbf{x}+\mathbf{n},
\end{equation}
where $\mathbf{x}=[h_t,\mathbf{q}_t,h_r,\mathbf{q}_r]^T$ is the composite channel coefficient vector, $\mathbf{n}=[n_1,...,n_\tau]^T$, and $\mathbf{V}$ is given at the bottom of this page. If $\mathbf{V}$ is of full rank, the LS estimate of $\mathbf{x}$ is given by $\hat{\mathbf{x}}=\sqrt{\frac{2}{p}}\mathbf{V}^{\dagger}\mathbf{y}$, with the minimum overhead $\tau\geq{2M+2}$. Then, the MSE of channel estimation for the ES protocol can be expressed as\cite{dft}

\setcounter{equation}{\value{equation}+1}

%\begin{equation}	
%	\mathbf{H}=\begin{gathered}
%		\begin{bmatrix}
%			s_{t,1} & s_{r,1}&\sqrt{{\beta _1^t}}{e^{j{\theta _{1,1}}}}s_{t,1}&...&\sqrt{{\beta _M^t}}{e^{j{\theta _{M,1}}}}s_{t,1}&\sqrt{{\beta _1^r}}{e^{j{\phi _{1,1}}}}s_{r,1}&...&\sqrt{{\beta _M^r}}{e^{j{\phi _{M,1}}}}s_{r,1} \\
%			\vdots&\vdots&\vdots&\ddots&\vdots&\vdots&\ddots&\vdots\\
%			s_{t,\tau} & s_{r,\tau}&\sqrt{{\beta _1^t}}{e^{j{\theta _{1,\tau}}}}s_{t,\tau}&...&\sqrt{{\beta _M^t}}{e^{j{\theta _{M,\tau}}}}s_{t,\tau}&\sqrt{{\beta _1^r}}{e^{j{\phi _{1,\tau}}}}s_{r,\tau}&...&\sqrt{{\beta _M^r}}{e^{j{\phi _{M,\tau}}}}s_{r,\tau}
%		\end{bmatrix}
%	\end{gathered}	
%\end{equation}

%By setting $\tau=2M+2$ and $s_{t,i}=s_{r,i}=1$, $\mathbf{H}$ is written as
%\begin{equation}	
	%	\mathbf{H}=\begin{gathered}
		%		\begin{bmatrix}
			%			1 & 1&\sqrt{{\beta _1^t}}{e^{j{\theta _{1,1}}}}s_{t,1}&...&\sqrt{{\beta _M^t}}{e^{j{\theta _{M,1}}}}s_{t,1}&\sqrt{{\beta _1^r}}{e^{j{\phi _{1,1}}}}s_{r,1}&...&\sqrt{{\beta _M^r}}{e^{j{\phi _{M,1}}}}s_{r,1} \\
			%			\vdots&\vdots&\vdots&\ddots&\vdots&\vdots&\ddots&\vdots\\
			%			s_{t,\tau} & s_{r,\tau}&\sqrt{{\beta _1^t}}{e^{j{\theta _{1,\tau}}}}s_{t,\tau}&...&\sqrt{{\beta _M^t}}{e^{j{\theta _{M,\tau}}}}s_{t,\tau}&\sqrt{{\beta _1^r}}{e^{j{\phi _{1,\tau}}}}s_{r,\tau}&...&\sqrt{{\beta _M^r}}{e^{j{\phi _{M,\tau}}}}s_{r,\tau}
			%		\end{bmatrix}
		%	\end{gathered}	
	%\end{equation}

%\begin{equation}
%	\hat{\mathbf{x}}=\mathbf{x}+\frac{\mathbf{H}^{-1}}{\sqrt{p/2}}\mathbf{n}.
%\end{equation}
%The covariance matrix of the CSI error is
\begin{equation}
	\text{MSE}^{\text{ES}}=\frac{2\sigma^2}{p}\text{Tr}[(\mathbf{V}^H{\mathbf{V}})^{-1}].
\end{equation}

\vspace{-0.5cm}
\subsection{Problem Formulation}

Define $\bar{\boldsymbol{\Theta}}$ and $\bar{\boldsymbol{\Phi}}$ as the training  pattern matrix that stacks the vector $\boldsymbol{\theta}_i^{\text{ES}}$ and $\boldsymbol{\phi}_i^{\text{ES}}$ of each time slot, respectively, i.e., $\bar{\boldsymbol{\Theta}}\triangleq[\boldsymbol{\theta}_1^{\text{ES}};...;\boldsymbol{\theta}_\tau^{\text{ES}}],\bar{\boldsymbol{\Phi}}\triangleq[\boldsymbol{\phi}_1^{\text{ES}};...;\boldsymbol{\phi}_\tau^{\text{ES}}]$.
For the ES protocol, we aim to minimize the MSE of channel estimation by jointly designing the pilot sequences of the users $\mathbf{s}_k{\triangleq\{s_{k,1},...,s_{k,\tau}\}^T},\forall{k\in\mathcal{K}}$, energy splitting ratio $\beta_m^k$, and training pattern matrices $\bar{\boldsymbol{\Theta}},\bar{\boldsymbol{\Phi}}$, which can be formulated as
\begin{subequations}\label{pro_es}
	\begin{align}
	{\rm{(P2)}}:\mathop {\min }\limits_{\{\mathbf{s}_k,\beta_m^k,\bar{\boldsymbol{\Theta}},\bar{\boldsymbol{\Phi}}\}} &\;\;\frac{2\sigma^2}{p}\text{Tr}[(\mathbf{V}^H{\mathbf{V}})^{-1}]\\
	\label{unit_mo_2}{\rm{s.t.}}&\;\;\theta_{m,i},\phi_{m,i}\in[0,2\pi),m\in\mathcal{M},i=1,...,\tau ,\\
	%\label{unit_pilot_2}&|s_{k,i}|=1,\forall{k\in\mathcal{K}},i=1,...,\tau,%\\
	\label{rank2}& {\rm{rank}}(\mathbf{V})=2M+2,\\
	&\beta_m^t+\beta_m^r\leq1,\\
	&\beta^k_m>0, k\in\mathcal{K},\\
	\label{couple2}&\cos({\theta_{m,i}}-{\phi_{m,i}})=0.
	\end{align}
\end{subequations}
%\vspace{-1cm}
%Problem (P2) is challenging as the variables are coupled in $\mathbf{V}$. Moreover, constraint (\ref{couple2}) makes the optimal solution difficult to find.
\subsection{Proposed Solutions}
To obtain useful insights, we first drop the constraints in (\ref{couple2}) and denote the relaxed problem as (P3). Note that (P3) corresponds to the ideal case, where the phase-shifts for transmission and reflection can be adjusted independently. 
\begin{proposition}
The optimal solution to problem (P3) should satisfy the following conditions: 
\begin{itemize}
	\item $\mathbf{V}$ is an orthogonal matrix.
	\item The pilot symbols of both users are always non-zero.
	\item The energy splitting ratio of each sub-surface is set identically as $\beta_m^t=\beta_m^r=0.5,\forall{m}\in\mathcal{M}$. 
\end{itemize}

\end{proposition}
\begin{proof}
The lower bound of MSE is attained when $\mathbf{V}^H\mathbf{V}$ is a diagonal matrix\cite{dft}, or in other words, the columns in $\mathbf{V}$ is orthogonal to each other. In this case, we have

\begin{equation}\label{proof1}
	\begin{aligned}	\text{Tr}[(\mathbf{V}^H{\mathbf{V}})^{-1}]&\overset{(a)}{=}{\sum_{q=1}^{2M+2}}\frac{1}{\sum_{p=1}^{\tau}|[\mathbf{V}]_{p,q}|^2}\\
	&\overset{(b)}{\geq}{\sum_{m=1}^{M}\frac{1}{\tau}(\frac{1}{\beta_m^t}+\frac{1}{\beta_m^r})}+\frac{2}{\tau}
	\end{aligned}
\end{equation}	
%where $[H]_{p,q}$ denotes each element in $\mathbf{H}$. 
where (a) is due to that the $q$-th diagonal element in $\mathbf{V}^H\mathbf{V}$ is $\sum_{p=1}^{\tau}|[\mathbf{V}]_{p,q}|^2$. If $\beta_m^k\neq0$, the equality of (b) will hold when the pilot symbols are all non-zero, i.e., $|s_{k,i}|=1,\forall{k,i}$. Finally, the MSE minimization is equivalent to minimizing $(\frac{1}{\beta_m^t}+\frac{1}{\beta_m^r})$, whose lower bound is achieved when $\beta_m^t=\beta_m^r=0.5$. This completes the proof.
\end{proof}

Next, we give an example of the optimal solution for problem (P3), which satisfies all the conditions in $\textbf{Proposition 1}$.
%introduce the design of pilot sequences and transmission and reflection pattern matrices to make $\mathbf{H}$ orthogonal. 
First, we select any ${(2M+2)}\times{(2M+2)}$ orthogonal matrix $\mathbf{D}_{2M+2}$, e.g., the DFT/Hadamard matrix. 
The pilot sequences $\mathbf{s}_t$ and $\mathbf{s}_r$ are set as the first and $(M+2)$-th column of $\mathbf{D}_{2M+2}$, respectively. Then, the training patterns $\bar{\boldsymbol{\Theta}},\bar{\boldsymbol{\Phi}}$ can be set as

\begin{subequations}\label{pattern1}
	\begin{align}			
		&[\bar{\boldsymbol{\Theta}}]_{m,n}=\frac{[\mathbf{D}_{2M+2}]_{m,n+1}}{s_{t,m}},1\leq{m}\leq{2M+2},1\leq{n}\leq{M}\\
		&[\bar{\boldsymbol{\Phi}}]_{m,n}=\frac{[\mathbf{D}_{2M+2}]_{m,n+M+2}}{s_{r,m}},1\leq{m}\leq{2M+2},1\leq{n}\leq{M}.
	\end{align}
\end{subequations}

%$\mathbf{H}$ in (13) can be rewritten as $[\mathbf{s}_t,{\text{diag}(\mathbf{s}_t)\bar{\boldsymbol{\Theta}}},\mathbf{s}_r,{\text{diag}(\mathbf{s}_t)\bar{\boldsymbol{\Phi}}}]$. 

The optimality of the above design can be explained as follows: First, for ease of description, $\mathbf{V}$ in (\ref{H_2}) can be written as
\begin{equation}\label{V2}
\mathbf{V}=	[\mathbf{s}_t,\frac{1}{\sqrt{2}}{\text{diag}(\mathbf{s}_t)\bar{\boldsymbol{\Theta}}},\mathbf{s}_r,\frac{1}{\sqrt{2}}{\text{diag}(\mathbf{s}_r)\bar{\boldsymbol{\Phi}}}].
\end{equation} 
Since we divide the phase-shifts by the pilot symbols in (\ref{pattern1}), %the columns of $[\text{diag}(\mathbf{s}_t)\bar{\boldsymbol{\Theta}},\text{diag}(\mathbf{s}_r)\bar{\boldsymbol{\Phi}}$ in $\mathbf{V}]$ are recovered as those in $\mathbf{A}_{\tau}$, only with an energy splitting ratio $\frac{1}{\sqrt{2}}$ multiplied on them. 
the $m$-th column ($\forall{n}\neq1,M+2$) in $\mathbf{V}$ equals the $m$-th column in $\mathbf{D}_{2M+2}$ times the energy splitting ratio $\frac{1}{\sqrt{2}}$. Meanwhile, the first and $(M+2)$-th column of $\mathbf{V}$ is the same as that in $\mathbf{D}_{2M+2}$ as introduced. Thus, $\mathbf{V}$ is an orthogonal matrix, which is an optimal solution to (P3). In this case, by substituting $\tau=2M+2$ and $\beta_m^k=0.5$ into (\ref{proof1}), the channel estimation MSE for ES is obtained as
\begin{equation}\label{mse_es_ideal}
	\text{MSE}^{\text{ES}}=(\frac{4M}{2M+2}+\frac{2}{2M+2})\frac{2\sigma^2}{p}=\frac{4M+2}{M+1}\frac{\sigma^2}{p},
\end{equation} 
which serves as a performance upper bound when evaluating the impact of practical phase-shifts.

%note that $\mathbf{H}$ can be written as $[\mathbf{s}_t,{\text{diag}(\mathbf{s}_t)\bar{\boldsymbol{\Theta}}},\mathbf{s}_r,{\text{diag}(\mathbf{s}_t)\bar{\boldsymbol{\Phi}}}]$

\underline{Proposed Solution to Problem (P2):} Due to constraint (\ref{couple2}) introduced by the coupled phase-shifts, it is hard to find an optimal solution for problem (P2). Therefore, we aim to find a high-quality suboptimal solution by constructing a \emph{nearly-orthogonal} matrix $\mathbf{V}$ under the full-rank constraint. 

Inspired by the transmission/reflection training design under the ideal phase-shift case, we target to retaining the orthogonality of  $[\text{diag}(\mathbf{s}_t)\bar{\boldsymbol{\Theta}},\text{diag}(\mathbf{s}_r)\bar{\boldsymbol{\Phi}}]$ in $\mathbf{V}$. Interestingly, we find that constraint (\ref{couple2}) can be met if the pilot sequence of R user is changed to $\mathbf{s}_r=[\jmath,-\jmath,\jmath,-\jmath,...]$
and the training pattern is designed as in (\ref{pattern1}). The reason is as follows: From (\ref{pattern1}), we can find that for  $1\leq{\forall{i}}\leq{2M+2},1\leq{\forall{m}}\leq{M},$
\begin{equation}\label{phase}
	\frac{e^{\jmath{\theta_{m,i}}}}{e^{\jmath{\phi_{m,i}}}}=\frac{[\mathbf{D}_{2M+2}]_{i,m+1}}{[\mathbf{D}_{2M+2}]_{i,m+M+2}}\frac{s_{r,i}}{s_{t,i}}.
\end{equation}
If $\mathbf{D}_{2M+2}$ is a DFT matrix, mathematically we have
\begin{equation}
[\mathbf{D}_{2M+2}]_{i,m+M+2}=(-1)^{i+1}[\mathbf{D}_{2M+2}]_{i,m+1}.	
\end{equation} 
Therefore, $\frac{e^{\jmath{\theta_{m,i}}}}{e^{\jmath{\phi_{m,i}}}}=\jmath$, which satisfies (\ref{couple2}).
If $\mathbf{D}_{2M+2}$ is the Hadamard matrix, since its entries are either $+1$ or $-1$, $\frac{e^{\jmath{\theta_{m,i}}}}{e^{\jmath{\phi_{m,i}}}}$ is either $\jmath$ or $-\jmath$ according to (\ref{phase}),  which also satisfies (\ref{couple2}). 
%Note that $\mathbf{s}_r$ may not be linearly independent with other columns, so we need $\tau\geq{2M+3}$ to ensure the full-rank constraint (\ref{rank2}) is satisfied.  
Based on the above, the columns of $\mathbf{V}$ are all orthogonal with each other except for one column, which is $\mathbf{s}_r$. Therefore, it is expected that the proposed STAR-RIS channel estimation scheme under the practical phase-shift model approaches the MSE performance of that under the ideal model.
%\begin{subequations}\label{con_prac}
	%	\begin{align}			
		%		\label{con_prac1}&s_{r,i}=\pm\jmath{s_{t,i}},\;1\leq{i}\leq{\tau},\\
		%		\label{con_prac2}&[\mathbf{A}_{\tau}]_{m,n+1+M}=\pm[\mathbf{A}_{\tau}]_{m,n+1},\;1\leq{m}\leq{\tau},1\leq{n}\leq{M}.
		%	\end{align}
	%\end{subequations}

%To ensure there is a unique solution to linear equation (\ref{signal_es}), $\mathbf{H}$ should at least satisfy $\text{rank}(\mathbf{H})={2M+2}$. Moreover, to guarantee a good performance of channel estimation, we can jointly design the pilot signals and phase-shifts to make most of the columns in $\mathbf{H}$ orthogonal while satisfying $\text{rank}(\mathbf{H})={2M+2}$.

\subsection{Discussion}
As introduced, the minimum overhead for channel estimation for TS and ES is the same, which is $2M+2$. 
Nevertheless, from (\ref{mse_ts}) and (\ref{mse_es_ideal}), we can observe that, the minimum channel estimation error for the ES protocol is approximately twice as that for the TS protocol. This can be intuitively explained by the fact that after energy splitting, part of the uplink signals is transmitted or reflected towards the opposite side of the STAR-RIS from the BS, which leads to a reduction in the effective signal strength during channel estimation. %To overcome this drawback, a hybrid protocol that enables channel estimation with TS and data transmission with ES is favorable for potential performance gain in STAR-RIS-aided networks.

%This phenomenon inspires us that in STAR-RISs-aided networks, the uplink transmission employing ES may suffer from performance degradation caused by the signal leakage. Furthermore, a harware design that enables hybrid TS-ES protocal for channel estimation and downlink transmission . 

%the uplink power loss caused by energy splitting protocol. More specifically, only about half of the trasmitted power are involved in the channel estimation process, while the other part are wasted since it can not reach the BS. 

%\subsection{Two-phase solution}
%Two phase: Estimate direct and cascaded links separately, which may cause error propagation. The performance is the same for ideal phase shift and practical phase shift. For the second phase of estimating the cascaded links, we can set phase shift as:
%\begin{equation}	
%	\mathbf{\Theta}=\begin{gathered}
%		\begin{bmatrix}
%			1 & 1&1&...&-q&-q&...&-q \\
%			1&\omega&\omega^2&\dots&\omega^{M}q&\omega^{M+1}q&\dots&\omega^{2M-1}q\\
%			\vdots&\ddots&\vdots&\dots&\vdots&\vdots&\ddots&\vdots\\
%			1 & \omega^{2M-1}&\omega^{2(2M-1)}&...&\omega^{M(2M-1)}q&\omega^{(M+1)(2M-1)}q&...&\omega^{(2M-1)(2M-1)}q
%		\end{bmatrix}
%	\end{gathered}	
%\end{equation}
%where $\omega=e^{-2\pi{q}/2M},q^2=-1$. This is based on the DFT matrix, where the odd-number rows of the last $M$ columns are multiplied by $-i$ while the even-number rows are multiplied by $i$. The pilot sequence of T user is 1 and $i$ for R user.

\section{Numerical Results}
In this section, we provide numerical results to verify the effectiveness of our proposed channel estimation schemes for TS and ES. In the simulation, the STAR-RIS consists of $M_0=80$ elements and is divided into $M=20$ sub-surfaces. All involved channels are modeled as Rician fading with the Rician factor of 10 dB. The distance-dependent path losses are modeled as $l=\beta_0(d/d_0)^{-\alpha}$, where $\beta_0=-30$ dB denotes the path loss at the reference distance $d_0=1$ meter (m), $d$ represents the individual link distance, and $\alpha$ is the path-loss exponent. We consider a two-dimensional coordinate system, where the BS is located at the origin and the reference center of the STAR-RIS is at (50m, 0). The T user and the R user are located at (54m, 3m) and (46m, -3m), respectively. The path-loss exponents of the user-BS, user-STAR-RIS, BS-STAR-RIS channels are set as $3.5$, $2.8$, and $2.2$, respectively. We set the noise power as $\sigma^2=-110$ dBm and the maximum transmit power of the user as $p=30$ dBm, unless otherwise stated. %The power spectrum density of AWGN is -170 dBm/Hz and the channel bandwidth is 1 MHz. All the results are averaged over 5000 independent channel realizations.

%In this paper, $h_k$,  $\mathbf{r}_k$, $\mathbf{g}$ are modeled as Rayleigh fading channels, i.e., $h_k\sim {\mathcal{C}\mathcal{N}}(0,l^{BU}_k),\mathbf{r}_k\sim {\mathcal{C}\mathcal{N}}(\mathbf{0},l^{SU}_k),\mathbf{g}\sim {\mathcal{C}\mathcal{N}}(\mathbf{0},l^{BS})$, where $l^{BU}_k,l^{SU}_k,l^{BS}$ denote the corresponding path losses. The distance-dependent path losses of $h_k$, $\mathbf{r}_k$ and $\mathbf{g}_k$ are modeled as $l^{BU}_k=\beta_0(d^{BU}_k/d_0)^{-\alpha_{BU}},l^{SU}_k=\beta_0(d^{SU}_k/d_0)^{-\alpha_{SU}},l^{BS}=\beta_0(d^{BS}/d_0)^{-\alpha_{BS}}$, where $\beta_0=-30$ dB denotes the path loss at the reference distance $d_0=1$ meter (m), $l^{BU}_k,l^{SU}_k,l^{BS}$ represent the distances from user $k\in\{t,r\}$ to the BS antenna, from user $k$ to the STAR-RIS and from STAR-RIS to the BS, respectively, and $\alpha_{BU},\alpha_{SU},\alpha_{BS}$ are corresponding path-loss exponents.  The path-loss exponents are set as $\alpha_{BU}=3.5,\alpha_{SU}=2.8,\alpha_{BS}=2.2$. The maximum transmit power $p$ is 30 dBm. The power spectrum density of AWGN is -170 dBm/Hz and the channel bandwidth is 1 MHz. All the results are averaged over 5000 independent channel realizations.
We compare the performance of the proposed channel estimation schemes with the following benchmarks:
\begin{itemize}
	\item \textbf{ON/OFF scheme for TS:} Following the idea in \cite{on_off}, the direct links are estimated with all sub-surfaces turned off and the cascaded links are estimated with one sub-surface turned on at transmission/reflection mode sequentially.
	\item \textbf{Two-phase channel estimation for ES:} In this scheme, the direct links and cascaded links are estimated separately in two phases. Specifically, in the first phase, only the users send orthogonal pilot sequences with the sub-surfaces off to estimate the direct links. In the second phase, the transmission and reflection patterns are set as (\ref{pattern1}) to estimate the cascaded channels.
\end{itemize}
\begin{figure}[t]
	\centering
	\includegraphics[width=0.42\textwidth]{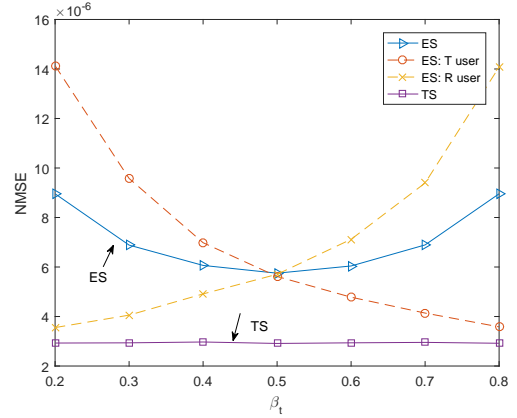}\\
	\caption{NMSE performance under different energy splitting ratios $\beta_t$.}
\end{figure}

\begin{figure}[t]
	\centering
	\includegraphics[width=0.42\textwidth]{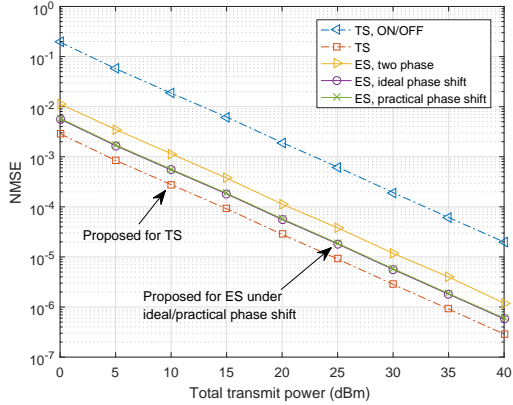}\\
	\caption{NMSE performance versus the total transmit power under different channel estimation schemes.}
\end{figure}
\begin{figure}[t]
	\centering
	\includegraphics[width=0.42\textwidth]{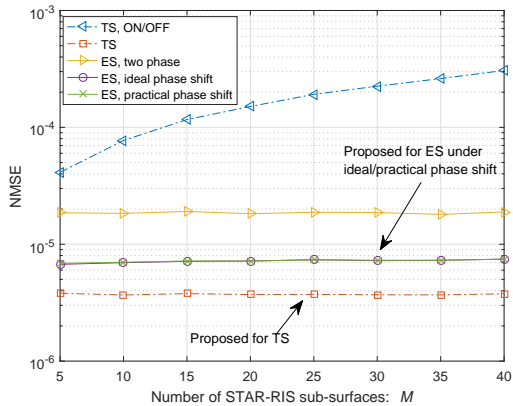}\\
	\caption{NMSE performance versus the number of sub-surfaces under different channel estimation schemes.}
\end{figure}

In Fig. 2, we plot the normalized MSE (NMSE) for STAR-RIS channel estimation under different energy splitting ratios. It is observed that by varying $\beta_t$, there exists a trade-off between the channel estimation accuracy of T and R users. Specifically, the minimum of the two users is obtained when $\beta_t=0.5$, which verifies our analysis in \textbf{Proposition 1}. Besides, with the optimal energy splitting ratio, the NMSE using ES is approximately twice as that using TS. This is because ES results in power leakage in the uplink channel estimation (see, Section IV. C).

In Fig. 3 and Fig. 4, we compare the performance of our proposed channel estimation schemes against the benchmarks. In Fig. 3, we plot the NMSE versus total transmit power. The key observations are made as follows: First, the NMSE decreases with the increasing of transmit power for all schemes, and the TS protocol yields the smallest NMSE. Second, the channel estimation error of the ON/OFF scheme is much larger than that of the other schemes since the large aperture of the surface is not fully utilized in the channel estimation stage. Third, the two phase-based scheme for ES behaves worse than our proposed scheme due to the error propagation issue. Specifically, the channel estimation error in the first phase for estimating direct links will deteriorate the performance in the second phase. Finally, the proposed scheme for ES with practical phase-shifts achieves close NMSE performance as that with ideal phase-shifts, which verifies the effectiveness of our proposed near-orthogonal training pattern design.

In Fig. 4, we examine the channel estimation performance of different schemes versus the number of sub-surfaces $M$. It is observed that the NMSE of the ON/OFF scheme increases with $M$ since the impact of noise accumulates with longer channel estimation time. For other schemes, the NMSE keeps constant, which is consistent with the analysis in (\ref{mse_ts}) and (\ref{mse_es_ideal}). Note that a larger $M$ requires longer channel estimation overhead, which leads to a shorter time of data transmission. Therefore, there exists a trade-off between the achievable rate and channel estimation overhead\cite{discrete_jsac}, which is an interesting topic in the future.

%It is observed that the NMSE of the ON/OFF scheme approximates that of the two-phase scheme since the training patterns become ill-conditioned when employing element-grouping. Moreover, the NMSE of the proposed schemes for ES and TS decreases as there is larger number of subsurfaces. Note that a larger $M$ requires longer channel estimation overhead, which leads to a shorter time of data transmission. Therefore, there exists a trade-off between the achievable rate and channel estimation overhead\cite{discrete_jsac}, which is an interesting topic in the future.

%In Fig. 4, we examine the channel estimation performance of different schemes versus the total transmit power. Owing to our effective designs, our proposed schemes for ES and TS achieve much better performance than the baseline schemes with the same transmit power. Besides, it is further verified from the figure that TS is a more power-saving protocol for uplink channel estimation in STAR-RIS-aided networks, compared with ES.

\section{Conclusion}
In this letter, we proposed efficient channel estimation schemes for a STAR-RIS assisted two-user communication system for the TS and ES protocols, respectively.
We first presented the optimal training patterns for TS to separately estimate the channels of two users. Then, a novel scheme for ES under the coupled phase-shift model was developed by jointly optimizing the pilot sequences, energy splitting ratio, and the training patterns, which achieves near-optimal MSE performance. Numerical results demonstrated that TS is more cost-effective than ES in terms of uplink channel estimation.

\bibliographystyle{IEEEtran}
\bibliography{ce_bib}

\end{document}